\documentclass[11pt]{article}

\usepackage{amssymb}
\usepackage{amsmath}
\usepackage{graphicx}
\usepackage{hyperref}

\newtheorem{proposition}{Proposition}[section]
\newtheorem{lemma}[proposition]{Lemma}
\newtheorem{observation}[proposition]{Observation}
\newtheorem{corollary}[proposition]{Corollary}
\newtheorem{theorem}[proposition]{Theorem}

\newtheorem{definition}[proposition]{Definition}

\newtheorem{remark}[proposition]{Remark}

\def\squarebox#1{\hbox to #1{\hfill\vbox to #1{\vfill}}}
\newcommand{\qed}{\hspace*{\fill}
       \vbox{\hrule\hbox{\vrule\squarebox{.667em}\vrule}\hrule}\smallskip}
\newenvironment{proof}{\begin{trivlist}
\item[\hspace{\labelsep}{\em\noindent Proof: }]
}{\qed\end{trivlist}}

\newlength{\tablength}
\newlength{\spacelength}
\newcommand{\tabstar}{\hspace*{\tablength}}
\newcommand{\spacestar}{\hspace*{\spacelength}}
\def\obeytabs{\catcode`\^^I=\active}
{\obeytabs\global\let^^I=\tabstar}
{\obeyspaces\global\let =\spacestar}
\newenvironment{display}{\begingroup\obeylines\obeyspaces\obeytabs}{\endgroup}
\newenvironment{prog}{\begin{display}\parskip0pt\sf}{\end{display}}
\newenvironment{pseudocode}{\smallskip\begin{quote}\begin{prog}}{\end{prog}\end
{quote}\smallskip}

\def\calM{{\cal M}}
\def\calP{{\cal P}}

\def\calS{{\cal S}}
\def\calU{{\cal U}}

\newcommand{\prob}[1]{\textsf{#1}}  

\newcommand{\notat}[1]{}  

\newcommand{\sched}{\prob{PC-Scheduling}}
\newcommand{\ssched}{\prob{PC-Capacity}}
\newcommand{\wcapacity}{\prob{PC-Weighted-Capacity}}

\begin{document}

\title{Wireless Scheduling with Power Control%
\thanks{A preliminary version of this paper
appears in the proceedings of the $17^{th}$ \textit{Annual European
  Symposium on Algorithms (ESA'09)}\cite{us:esa}.}
}

\date{\today}


\author{
Magn\'{u}s M.\ Halld\'{o}rsson\thanks{ICE-TCS,
School of Computer Science,
Reykjavik University, 
101 Reykjavik, Iceland.
Work done in part while 
visiting Research Institute for Mathematical Sciences (RIMS) at
Kyoto University. Supported by Icelandic Research Fund
grant 90032021. Email: \url{mmh@ru.is}
} 
}

\maketitle

\begin{abstract}
We consider the scheduling of arbitrary wireless links in the
physical model of interference to minimize the time for
satisfying all requests.
We study here the combined
problem of scheduling and power control, where we seek both
an assignment of power settings and a partition of the links so that
each set satisfies the signal-to-interference-plus-noise (SINR) constraints.

We give an algorithm that attains an approximation ratio of $O(\log n
\cdot \log\log \Delta)$, where $n$ is the number of links and $\Delta$ 
is the ratio between the longest and the shortest link length. 
Under the natural assumption that
lengths are represented in binary, this gives the first
approximation ratio that is polylogarithmic in the size of the input.
The algorithm has the desirable property
of using an oblivious power assignment, where the power assigned to a
sender depends only on the length of the link. We give evidence that
this dependence on $\Delta$ is unavoidable, showing that any 
reasonably-behaving oblivious power assignment results in a $\Omega(\log\log
\Delta)$-approximation.


These results hold also for the (weighted) capacity problem of
finding a maximum (weighted) subset of links that can be scheduled in
a single time slot. In addition, we obtain improved approximation for a
bidirectional variant of the scheduling problem, give partial
answers to questions about the utility of graphs for modeling physical
interference, and generalize the setting from the standard
2-dimensional Euclidean plane to doubling metrics.
Finally, we explore the utility of graph models in capturing wireless 
interference.
\end{abstract}

\section{Introduction}

We are interested in fundamental limits on communication in wireless
networks. How much communication throughput is possible?
This is an issue of efficient spatial separation, keeping the interference
from simultaneously communicating links sufficiently low.
The interference scheduling problem is then to schedule an arbitrary
set of communication links in the least amount of time while 
satisfying interference constraints. In this paper, we focus on the
power control version, where we also choose the power settings for
the links. 

The scheduling problem depends strongly on the model of
interference. Until recently, previous algorithmic work has revolved
around various graph-based models, where interference is modeled as a
pairwise constraint. This, however, fails to capture the accumulative
property of actual radio signals. 
In contrast, researchers in information, communication, or
network theory (``EE'') are working with wireless models that sum
up interference and respect attenuation. The standard model is the
signal-to-interference-plus-noise (SINR) model, to be formally introduced 
in Section \ref{sec:notation}. The SINR model reflects 
physical reality more accurately and is therefore often simply
called the physical model.
On the other hand, most research in the SINR model has focused on 
heuristics that are evaluated by simulation, which
neither give insights into the complexity of the problem
nor give algorithmic results that may ultimately lead to
new protocols. 

Formally, given is an arbitrary set of links, each a
sender-receiver pair of points in the plane. We seek an assignment of
power settings to the senders and a partition of the linkset into minimum
number of slots, so that the set of links in each slot satisfies 
the SINR-constraints.
We refer to this as the {\sched} problem.
We also consider two closely related throughput maximization problems, both with power control. 
In the {\ssched} problem, we seek a maximum cardinality subset of links satisfying the SINR constraints, while in the {\wcapacity} problem, the links have given weights and we seek to maximize the total weight of a feasible subset.
Finally, we also touch on the bidirectional setting, 
where both nodes in a link may be
transmitting, implying a stronger, symmetric form of interference.

For reasons of simplicity of use, it is strongly desirable to use power
assignments that are precomputable independent of other links. 
Such \emph{oblivious} assignments depend only on
the length of the given link. In fact, oblivious assignments
appear essential in the distributed setting.
The two most frequently used power
assignment strategies are indeed of this type, using either \emph{uniform}
(or fixed) power for all the links, or \emph{linear} assignment that
ensures that the signals received at the intended receivers are identical.

The other issue of particular interest is the utility of graphs
for modeling interference. It is clear that graphs
are imperfect models, given both the non-locality and the additive nature of
interference in the SINR model. The perceived difficulty in reasoning
analytically about these additional complications has been cited as a
factor against SINR model. Still, graphs have proved to be highly
versatile tools for analysis and algorithm design, and pairwise
constraints are in general much easier to handle than many-to-many
constraints. We would therefore like to quantify the cost of doing business
using graphs, or the overhead that amenable graph models have over
non-graphic models, as well as pinpointing particular situations where
graphs work especially well.

\subsection{Our Contributions}

We give upper and lower bounds on the quality of oblivious power assignments
for wireless scheduling problems with power control.
We obtain algorithms for all three problems
that attain a $O(\log\log \Delta \cdot \log n)$-approximation, 
using a recently introduced oblivious 
\emph{mean} (or \emph{square root} \cite{FKV09}) power assignment.
This is an exponential improvement over previous results in terms of $\Delta$,
and leads to a polylogaritmic approximation ratio in terms of the
length of the input (under the natural assumptions that lengths be
represented in binary).
This dependence on $\Delta$ turns out to unavoidable ---
we show that any reasonable oblivious power function forces
$\Omega(\log\log \Delta)$-approximate schedules,

In the bidirectional setting, we obtain a 
$O(\log n)$-approximation, improving on the previous 
$O(\log^{c}n)$-factor for {\sched} with $c > 5$ \cite{FKRV09} 
using considerably simpler arguments.

We precede this analysis with a study of the applicability of uniform
power, as a form of ultra-oblivious power assignment, tying together
a number of known results.
Namely, $O(\log \Delta)$-ratios can be attained online or by
distributed algorithms. This had previously only been stated
explicitly for {\ssched}. 
We additionally extend the current state-of-the-art in two ways.

We generalize the setting from the plane to the class of
\emph{doubling metrics}.
This assumes that path-loss constant $\alpha$ is greater than the
doubling constant of the metric, which is equivalent 
to the standard assumption that $\alpha > 2$ in the plane (see Section
\ref{sec:notation}). This is to assume 
that the cumulative power of a transmission
fades away, and dub this combination of metric and path-loss constant as a
\emph{fading} metric.

Our work aims also to address the utility of graphs in representing
physical models of interference, and our results indicate that even if
imperfect as models, graphs can still play a useful role.  In particular,
for links of nearly equal length, we show that unit-disc graphs
capture the SINR-constraints, within constant factors, reducing the
problems to the well-studied (weighted) independent set and coloring
problems.  The $O(\log n \cdot \log\log \Delta)$-approximation result
is also relative to the underlying graph.

The current paper refines the results and the arguments in the
earlier draft \cite{us:esa}, and adds to it approximations of
\wcapacity. Additionally, the draft \cite{us:esa} contained a faulty
lemma (Lemma 4.3), which is corrected here by proving the main results
in Section \ref{sec:oblivious} differently. The new treatment involves an
extension of a graph-theoretic approximation property, which may be of
independent interest.

\subsection{Related Work}
\label{sec:related-work}

Most work in wireless scheduling in the physical (SINR) model has been
of heuristic nature, e.g. \cite{ElBattE04journal}. Only
after the work of Gupta and Kumar \cite{kumar00} did analytical
results became \emph{en vogue}, but
were largely non-algorithmic and restricted to networks with a
well-behaving topology and traffic pattern such as uniform geometric
distribution. 

In contrast, the body of algorithmic work is mostly on graph-based models
that ultimately abstract away the nature of wireless communication.
The inefficiency of graph-based protocols in the SINR model is
well documented and has been shown both theoretically and
experimentally~\cite{GronkMibiHoc01,MaheshwariJD08,Moscibroda2006Protocol}.

Algorithmic work in the SINR model started in 2006 with the seminal work of Moscibroda and Wattenhofer \cite{MoWa06}. 
In this paper, Moscibroda and Wattenhofer present an algorithm that
successfully schedules a set of links (carefully chosen to
strongly connect an arbitrary set of nodes) into polylogarithmic
number of slots, even in arbitrary worst-case networks. In contrast to our
work, the links themselves are \emph{not} arbitrary (but do have
structure that will simplify the problem). This work has been
extended and applied to topology
control~\cite{gao08,moscibroda06b}, sensor
networks~\cite{Moscibroda07}, and combined scheduling and
routing~\cite{chafekar07}.
However, arbitrary networks are beyond the scope of these papers.
Apart from these papers, algorithmic SINR results
also started showing up here and there, for instance in a game
theoretic context or a distributed algorithms context,
e.g.,
\cite{AulMosPenPer08,AEK08,BrarBS08,Goussevskaia2008Local,ScheidelerRS08,KatzVW08}.


Approximation algorithms for the problem of scheduling wireless links
with power control in the SINR model were given in \cite{moscibroda06b},
\cite{MoscibrodaOW07} and \cite{chafekar07}.
In all cases the performance ratios obtained consist of the product of
structural properties and a function of the number of nodes.
The structural properties are different but can all grow linearly with
the size of the network.

A number of recent related results have featured a $O(\log
\Delta)$-like approximation in the plane (assuming $\alpha > 2$).
Goussievskaia, Oswald and Wattenhofer
\cite{gouss2007} gave a $O(\log \Delta)$-factor approximation for
both the scheduling and the (weighted) capacity problem. They compared
their algorithm to the optimal solution constrained to use uniform
power assignment, but it requires only a small step to relate it to
optimum with power control.
Andrews and Dinitz \cite{AD09} applied this extra step to obtain a $O(\log
\Delta)$-approximation for {\ssched}.
Fangh\"anel, Kesselheim and V\"ocking \cite{FKV09}
used a different approach and gave a randomized algorithm for {\sched} 
that uses $O(OPT \log \Delta + \log^2 n)$ 
slots. 
Finally, Avin, Lotker and Pignolet \cite{ALP09} show that 
the assumption of $\alpha > 2$ used by all previous work may not be
necessary, in that the ratio between optimal non-oblivious and oblivious
capacity is $O(\log \Delta)$, at least in the 1-dimensional metric. 

In \cite{FKRV09}, Fangh\"anel et al.\ gave a construction that 
shows that any schedule based on any oblivious power assignment 
can be a factor of $n$ from optimal.
They also introduced the bidirectional version of
the scheduling problem and give a $O(\log^{3.5+\alpha} n)$-approximation factor
using the mean power assignment in general metrics. 
Their proof involves non-trivial embeddings into tree metric spaces.

In contrast, the scheduling complexity of arbitrary links in the case of fixed,
uniform power is better understood.
Constant factor
approximation for the corresponding capacity
problem in the plane was given in \cite{GHWW09}, 
yielding a $O(\log n)$-approximation for the scheduling problem.
Both of these problems are known to be NP-complete \cite{gouss2007}.
The results obtained here for power control build on and extend the techniques
and properties derived in the case of uniform power in \cite{GHWW09,HW09}.

In developments since the original presentation of this work \cite{us:esa},
Erlebach and Grant \cite{ErlebachGrant} gave a $O(\log \Delta)$-factor
approximation algorithm for the problem of multicast scheduling, where
each transmission is to be sent to a collection of receivers.
Their work uses in a fundamental way the results of the current paper 
on nearly-equilength links and unit-disc graphs.
Fangh{\"a}nel et al.~\cite{FanghanelOnline2010} studied the online
version of {\ssched} problem, obtaining a tight bound of $\theta(\Delta^{d/2})$
on the competitive ratio of deterministic algorithms in
$d$-dimensional Euclidean space. 

In a breakthrough, Kesselheim \cite{kesselheimSODA11} has very
recently obtained a $O(1)$-approximation algorithm for {\ssched}.
It necessarily uses instance-specific power assignment, and the
question of optimal schedules using oblivious power assignment remains
interesting both from a theoretical and practical viewpoint.
Halld\'orsson and Mitra \cite{HallMitra11} have generalized our
results for {\ssched} to arbitrary metric spaces. They additionally
obtained the improved approximation factors of $O(\log n + \log\log \Delta)$
and $O(1)$ in the uni-directional and bi-directional cases, respectively.
For {\sched} with oblivious power, however, 
ours are still the best approximation factors known.

\section{Notation and Preliminaries}
\label{sec:notation}

Given is a set $L = \{\ell_1, \ell_2, \ldots, \ell_n\}$ of links, where
each link $\ell_v$ represents a communication request from a sender
$s_v$ to a receiver $r_v$. \label{def:lv}
The distance between two points $x$ and $y$ is denoted $d(x,y)$.
\label{def:dxy}
The asymmetric distance from link $\ell_v$ to link $\ell_w$ is the distance from
$\ell_v$'s sender to $\ell_w$'s receiver, denoted $d_{vw} = d(s_v, r_w)$.
\label{def:dvw}
The length of link $\ell_v$ is denoted 
simply $\ell_v$.
We shall assume for simplicity of exposition that all links are of
different length; this does not affect the results materially.
We assume that each link has a unit-traffic demand, and model the case
of non-unit traffic demands by replicating the links.

The nodes can transmit with different power. 
Let $P_v$ denote the power assigned to link $\ell_v$.
\label{def:pv}
%
We assume the \emph{path loss radio propagation} model for the
reception of signals, where the signal received from $s_w$ at receiver
$r_v$ is $P_w/d_{wv}^\alpha$ and $\alpha$ denotes the
path-loss exponent. 
\label{def:alpha}
%
We adopt the \emph{physical interference model}, in which a node $r_v$
successfully receives a message from a sender $s_v$ if and only if the
following condition holds:
\begin{equation}
 \frac{P_v/\ell_v^\alpha}{\sum_{\ell_w \in S \setminus  \{\ell_v\}}
   P_w/d_{wv}^\alpha + N} \ge \beta, 
 \label{eq:sinr}
\end{equation}
where $N$ is the ambient noise, $\beta$ denotes the minimum
SINR (signal-to-noise-ratio) required for a message to be successfully received,
and $S$ is the set of concurrently scheduled links in the same \emph{slot}.
\label{def:beta}
Note that by scaling the power of
all the senders, the effect of the noise $N$ can be made arbitrarily
small, thus we ignore this term.
Of course, in real situations, there are upper bounds on maximum
power which we ignore here.
We shall also assume that $\beta\ge 3^\alpha$; 
by the signal-strengthening results of \cite{HW09}, 
this can only affect the constants in the approximation results.
We say that $S$ is \emph{SINR-feasible} if (\ref{eq:sinr}) is
satisfied for each link $\ell_v$ in $S$.

This paper deals with \emph{power control}, i.e., determining the
power assignment to the links is a part of the problem. In particular,
we focus on \emph{oblivious power} assignments, where the power depends
only on the length of the link, while we compare it to an optimal
solution that is free to use any power assignment. 
The most basic assignment is \emph{uniform power}, where each link
$\ell_v$ uses the same power $P_v = P$. Another common oblivious
assignment is \emph{linear power}, where $P_v = \ell_v^\alpha$.
We will focus on uniform power, along with another oblivious
assignment, the \emph{mean} (
(or, square-root \cite{FKV09}) power $\calM$ given by $P_v = \calM_v =
\ell_v^{\alpha/2}$.
\label{def:calM}

The \emph{affectance} of link $\ell_v$ caused by a set $S$ of
links \cite{GHWW09,HW09} under a given power assignment $P$, 
is the sum of the interferences of the links in $S$ on
$\ell_v$ relative to the power received, or
  \[ a_{S}(\ell_v) 
     = \sum_{\ell_w \in S \setminus \{\ell_v\}} \frac{P_w/d_{wv}^\alpha}{P_v/\ell_v^\alpha}
     = \sum_{\ell_w \in S \setminus \{\ell_v\}} \frac{P_w}{P_v} \cdot \left(\frac{\ell_v}{d_{wv}}\right)^\alpha
  \] 
\label{def:asv}
For a single link $\ell_w$, we use the shorthand $a_{w}(v) = 
a_{\{\ell_w\}}(\ell_v)$.
\label{def:awv}
Note that affectance is additive in that for disjoint sets of links
$S_1, S_2$, $a_{S_1 \cup S_2}(\ell_v) = a_{S_1}(\ell_v) + a_{S_2}(\ell_v)$.

A \emph{$p$-signal} set or a schedule is one where the
affectance of any link is at most $1/p$, with respect to the given power assignment.
\label{def:psignal}
A set is SINR-feasible iff it is a $1/\beta$-signal set, i.e., 
$a_S(\ell_v) \le 1/\beta$, for each link $\ell_v \in S$.
Let $OPT_p$ be a $p$-signal schedule with minimum number of slots.
\label{def:optp}
Let $\Delta$ denote the ratio between the maximum and minimum length of a link. 
\label{def:delta}


For a graph $G$, let $\chi(G)$ denote its chromatic number, 
and $\alpha(G)$ its independence number (or the maximum cardinality of
a subset of mutually non-adjacent vertices).
\label{def:chi}
\label{def:alphag}
Define the neighborhood $N(v)$ of a vertex $v$ to be the set
consisting of $v$'s neighbors, and the closed neighborhood $N[v]$ to include $v$ as well. 
\label{def:openneigh} 
\label{def:closedneigh} 
For a vertex subset $S$, let $G[S]$ denote the subgraph induced by $S$. 
\label{def:induced}

\smallskip

We say that a collection of links is 
\emph{$q$-independent} 
if any two of them, $\ell_v$ and $\ell_w$, 
satisfy the constraint 
 \[ d_{vw} \cdot d_{wv} \ge q^2 \cdot \ell_w \ell_v\ . \]
\label{def:qind} 
Define the \emph{link graph} $G_q(L)$ 
on a link set $L$, parameterized by a constant $q$, such that 
a pair of links are adjacent in $G_q$ iff they are
not $q$-independent.
\label{def:gq} 

The following observation shows that a schedule of a linkset forms a
coloring of the corresponding link graph. The converse, however, does
not necessarily hold, as we shall see. Thus, the graph representation
is more relaxed than required.

\begin{table}
\begin{tabular}{llcl}
\emph{Notation} & \emph{Meaning} & \emph{Topic} & \emph{Page} \\
\hline
$\ell_v$ & Link $\ell_v = (s_v,r_v)$; denotes also its length & & \pageref{def:lv} \\
$P_v$ & Power assigned to link $\ell_v$ & & \pageref{def:lv} \\
$\calM_v$ & Mean power assignment $\calM_v = \ell_v^{\alpha/2}$ & &
\pageref{def:calM} \\
$\alpha$ & Path loss constant (signal decay exponent). & \emph{SINR} & \pageref{def:alpha} \\
$\beta$ & SINR requirement (assumed to be at least $3^\alpha$). && \pageref{def:beta} \\
$d(x,y)$ & Distance between points $x$ and $y$. & & \pageref{def:dxy} \\
$d_{vw}$ & $= d(s_v, r_w)$ & & \pageref{def:dvw} \\
\hline

$a_{S}(\ell_v)$ & Affectance of linkset $S$ on link $\ell_v$ & & \pageref{def:asv} \\
$a_{w}(v)$ & $=a_{\{\ell_w\}}(\ell_v)$ & & \pageref{def:awv} \\
$OPT_p$ & Optimal $p$-signal schedule & \emph{Analysis} & \pageref{def:optp} \\
$\Delta$ & Ratio of longest to shortest link length & & \pageref{def:delta} \\
$\zeta(x)$ & Riemann zeta-function. & & \pageref{def:zeta} \\
\hline

$q$-independent & $d_{vw}\cdot d_{wv} \ge q^2 \cdot \ell_v \ell_w$ & & \pageref{def:qind}  \\
$t$-close & $\max(a_v(w),a_w(v))\ge t$ & & \pageref{def:tclose} \\
well-separated & Link Lengths differ by factor $\le 2$ or $\ge
\Lambda $ & \emph{Link relationships} & \pageref{def:wellsep} \\
$\tau$ & $2\beta n$ & & \pageref{def:tau} \\
$\Lambda$ & $2 \tau^{1/\alpha}$ & & \pageref{def:lambda} \\
\hline

$\chi(G)$ & Chromatic number of graph $G$ & & \pageref{def:chi} \\
$\alpha(G)$ & Independence number of graph $G$ & & \pageref{def:alphag} \\
$G_q(L)$ & $q$-independence relation on linkset $L$ & & \pageref{def:gq} \\
$U_z(L)$ & Unit-disc graph on the senders in $L$ & \emph{Graphs} & \pageref{def:uz} \\
$G[X]$ & Graph induced by vertex subset $X$ & & \pageref{def:induced} \\
$N_G(v)$ & Set of neighbors of node $v$ in graph $G$ & & \pageref{def:openneigh} \\
$N_G[v]$ & Closed neighborhood of $v$, $=N_G(v)\cup \{v\}$ & & \pageref{def:closedneigh} \\
\hline

$A = \dim_A(\calU, d)$ & The Assouad (doubling) dimension (2, for $\mathbb{R}^2$) & & \pageref{def:a} \\
$B(y,\epsilon)$ & Ball of radius $\epsilon$ centered at $y$. & \emph{Metrics} & \pageref{def:b} \\
$C$ & Constant in doubling dimension definition & & \pageref{def:c} \\
\hline

$C'$ & $=\alpha C 4^A \zeta(\alpha+1-A)$ & & \pageref{def:cprime} \\
$z_1 = z_1(p)$ & $=4(pC')^{1/\alpha}$. Sufficient sender separation. & & \pageref{def:z1} \\
$z_2 = z_2(p)$ & $=p^{1/\alpha}-1$. Necessary sender separation. & & \pageref{def:z2}  \\
\hline

\end{tabular}
\label{table:notation}
\caption{List of notation}
\end{table}

\begin{lemma}
If $S$ is a $q^\alpha$-signal set under some power assignment, 
then $S$ is $q$-independent.
\label{lem:ind-separation}
\end{lemma}

\begin{proof}
Let $P$ be a power assignment for which $S$ is a $q^\alpha$-signal set.
Since the links belong to the same $p$-signal set, for $p=q^\alpha$,
they satisfy
\[ \frac{P_v/\ell_v^\alpha}{P_w/d_{wv}^\alpha} \ge p, ~~~\mbox{and}~~~
   \frac{P_w/\ell_w^\alpha}{P_v/d_{vw}^\alpha} \ge p\ . \]
By multiplying these inequalities together and rearranging, we get that
$d_{vw} \cdot d_{wv} \ge p^{2/\alpha} \cdot \ell_w \ell_v
         = q^{2} \cdot \ell_w \ell_v$.
\end{proof}

\section{Approximations Using Uniform Power}

One of the most widely used power assignment is the uniform one, where
senders use the same power setting. This might be viewed as
ultra-oblivious, as transmissions are now independent of link length.

In Sec.~\ref{sec:nearlyequi}, we show that uniform power assignment performs 
well when links are of nearly equal lengths. 
The global nature of the
problem disappears, and local strategies become sufficient.
This results in $O(\log \Delta)$-approximation algorithms 
using any oblivious power assignment, which is argued in Sec.~\ref{sec:arbitrary}.
In fact, the algorithms for {\sched} and {\ssched} are $O(\log
\Delta)$-competitive online algorithms.
These results essentially follow with minor effort from previous
works, in particular \cite{gouss2007}.
%

The new contributions in this section are twofold.
We introduce fading metrics in Sec.~\ref{sec:fading} and show that a
well-dispersed set of links in such a metric has good signal properties.
We also show in Sec.~\ref{sec:nearlyequi} that unit-disc graphs capture
well links of nearly equal lengths.

\subsection{Scheduling in Fading Metrics}
\label{sec:fading}

We extend the traditional setting from the Euclidean plane to doubling metrics (see Clarkson \cite{Clarkson}). 

A \emph{metric space} is a pair $(\calU, d)$, where $\calU$ is a set and $d$ is a distance function, satisfying: $d(x,x)=0$, $d(x,y) = d(y,x)$ (symmetry), and
$d(x,y) + d(y,z) \le d(x,z)$ (triangular inequality), for any points
$x,y,z \in \calU$. 
Intuitively, a metric space is \emph{doubling} if the volume of a ball
increases by at most a constant times the radius. Let $B(y,\epsilon) =
\{x \in \calU | d(x,y) < \epsilon\}$ be the \emph{$\epsilon$-ball} centered
at $y$. 
\label{def:b} 
A set $Y \subset \calU$ is an \emph{$\epsilon$-packing} if $d(x,y) >
2\epsilon$, for any $x,y \in Y$. That is, the set of balls
$\{B(y,\epsilon) | y \in Y\}$ are disjoint.  The packing number
$\calP(\calU,\epsilon)$ is the size of the largest $\epsilon$-packing,
i.e., the maximum number of $\epsilon$-balls that can be packed into
the body $\calU$.
%
The \emph{Assouad dimension} $\dim_A(\calU, d)$ \cite{Assouad}
(also known as uniform
metric dimension or doubling dimension) for a metric space
($\calU, d$) is the value $t$, if it exists, such that
\[ \sup_{x \in \calU, r > 0} \calP(B(x,r), \epsilon r) = C \cdot 1/\epsilon^{t}, \]
as $\epsilon \rightarrow 0$, where $C$ is an absolute constant.
\label{def:a} 
\label{def:c} 
It is known that $\dim_A(\Re^k) = k$ for the $k$-dimensional Euclidean
space \cite{Heinonen}, and in particular for the plane 
$C = \frac{1}{6}\pi \sqrt{3} \approx 0.907$ \cite{Toth,Weisstein}.

We require that the
path loss exponent $\alpha$ be strictly greater than the doubling dimension
$A = \dim_A(\calU,d)$ of the metric. This requirement is the reason
for not using simpler dimension definitions that are
equivalent only up to a constant factor.
We shall refer to such a
combination of distance metric and path loss constant as a
\emph{fading metric}.

The following result extends similar lemmas in previous works (see
\cite{GHWW09,HW09}) from the setting of the Euclidean plane to the
more general class of fading metrics.  It yields a converse of Lemma
\ref{lem:ind-separation} for the case of nearly-equilength links.
This is the only place where we use the fading property of the metric,
i.e., that $\alpha$ is strictly greater than the doubling dimension.

Let $\zeta(x) = \sum_{t \ge 1} \frac{1}{t^x}$ be the Riemann zeta-function,
which is well-defined for any $x > 1$.
\label{def:zeta} 
Let $C' = \alpha C 4^A \zeta(\alpha+1-A)$ and 
let $z_1(p) = 4 (p C')^{1/\alpha}$.
\label{def:z1} 
\label{def:cprime} 

\begin{lemma}{(Far-away lemma)}
Let $p$ be positive and 
let $S$ be a set of links whose senders are of mutual distance at least
$z D/2$, where $D$ is the length of the longest link in $S$ and $z = z_1(p)$.
Then, using uniform power assignment, $S$ forms a
$p$-signal set in any fading metric.
\label{lem:faraway-aff}
\end{lemma}

\begin{proof}
Let $S'$ be the set of senders of links in $S$.
Let $Z = zD/4$.
The separation of the senders implies that $S'$ is a $Z$-packing.
The definition of a doubling metric implies that
for any $t > 0$, the packing number of the $tZ$-ball centered at any point $x$ is
bounded by 
\begin{equation}
 \calP(B(x,t Z),Z) \le C t^A\ .
\label{eq:packing-bnd}
\end{equation}
Namely, any packing of balls of radius $Z$ inside a ball of radius
$tZ$ contains at most $C t^A$ balls.

Let $g$ be a number.
Let $s_x$ be a sender in $S'$ belonging to link $\ell_x$.
Let $S_g = \{ s_y \in S' | d(s_x,s_y) < gZ \}$ be the set of senders
within distance less than $gZ$ from $s_x$, 
and let $T_g = S_g \setminus S_{g-1}$.
By assumption, $S_2 = \emptyset$.
Each sender $s_y$ in $T_g$ is of distance at least $(g-1)Z$ from $s_x$,
so $d_{yx} \ge (g-1)Z - D \ge (g-2)Z$.
Since $\ell_x \le D$, the affectance of $\ell_y$ on $\ell_x$ is at most
\[ a_y(x) = \frac{1/d_{yx}^\alpha}{1/\ell_x^\alpha} 
    \le \left(\frac{D}{(g-2)Z} \right)^{\alpha}
     = \left(\frac{4}{(g-2)z} \right)^{\alpha}, \quad \forall \ell_y\in T_g . \]
Observe that
\[ \frac{1}{(g-1)^\alpha} - \frac{1}{g^\alpha} 
  = \frac{g^\alpha - (g-1)^\alpha}{g^\alpha (g-1)^\alpha} 
 \le \frac{\alpha g^{\alpha-1}}{g^\alpha (g-1)^\alpha} 
  < \frac{\alpha}{(g-1)^{\alpha+1}}\ .  \]
Then,
\begin{align}
 a_S(x) & = \sum_{g\ge 3} a_{T_g}(x) \nonumber \\\nonumber
   & \le \sum_{g\ge 3} |S_g \setminus S_{g-1}|\cdot \left(\frac{4}{(g-2)z}\right)^\alpha \\\nonumber
   & = \left(\frac{4}{z}\right)^\alpha
       \sum_{g\ge 3} |S_g| \left(\frac{1}{(g-2)^\alpha} -
       \frac{1}{(g-1)^\alpha} \right) \\
   & \le \left(\frac{4}{z}\right)^\alpha
           \sum_{g\ge 3} |S_g| \frac{\alpha}{(g-2)^{\alpha+1}} \ . \label{eq:mid-faraway}
\end{align}
The balls of radius $Z$ centered at points in $S_g$
are all contained within the ball $B(x,(g+1)Z)$.
For $g\ge 3$, the packing bound (\ref{eq:packing-bnd}) then implies that
$|S_g| \le \calP(B(x,(g+1)Z),Z) \le C (g+1)^A$, and thus we have that
\begin{align*}
\frac{|S_g|}{(g-2)^{\alpha+1}} 
  \le \frac{C (g+1)^A}{(g-2)^{\alpha+1}}
  \le \frac{4^A C}{(g-2)^{\alpha+1-A}} \ .
\end{align*}
Continuing from (\ref{eq:mid-faraway}),
\[ a_S(x) \le \left(\frac{4}{z}\right)^\alpha\alpha\cdot 4^AC 
    \sum_{x \ge 1} \frac{1}{x^{\alpha+1-A}} = 
  \left(\frac{4}{z}\right)^\alpha\alpha\cdot 4^AC
  \zeta(\alpha+1-A) = \left(\frac{4}{z}\right)^\alpha C'\ . \]
Thus, $S$ is an $s$-signal set, where
$s = \frac{1}{C'} (\frac{z}{4})^{\alpha} = p$.
\end{proof}

\begin{remark}
Lemma \ref{lem:faraway-aff} does not hold in general for arbitrary
distance metrics. In particular, it fails for $\mathbf{R}^2$ when
$\alpha \le 2$.
In fact, unit-length links arranged in a grid with separation $q$
will be $q^\alpha$-independent, while maximum affectance becomes
$\Omega(\log n)$. 
\end{remark}

\subsection{Modelling Nearly-Equilength Links as Unit-Disc Graphs}
\label{sec:nearlyequi}

We observe here that if the links are of nearly equal length, then we
can simplify the many-to-many interference relationships by a pairwise
relationship, modulo small constant factors in the approximation.
These pairwise relationships correspond to the graphs formed by
discs of fixed radius in the plane. With one radius, we capture the
necessary distance between any pair of links in a feasible solution,
while with another larger radius, we have the sufficient distance so
that any set of links of such mutual separation is guaranteed to be
SINR-feasible (in any given fading metric). This leads to simple and effective
approximation algorithms that can be made online and turned into
distributed algorithms.


We say that a set of links is \emph{nearly-equilength} if lengths of
any pair of links in the set differ by a factor of less than 2.
The key observation is that we can represent the link graph $G_q =
G_q(L)$ of a set $L$ of nearly-equilength links approximately with a
unit-disc graph (UDG).

\begin{definition}
Let $L$ be a linkset in a fading metric, and let $d$ denote the
minimum link length in $L$. 
Given a number $z$,
the \emph{unit-disc graph} $U_z(L)$ of $L$ is the graph with a node for
each sender of $L$ with two nodes adjacent if the distance between
the two senders is less than $z \cdot d$.
\label{def:uz}
\end{definition}

That is, $U_z(L)$ is the graph formed by the intersection of balls
of radius $zd/2$ that are centered at the senders.
We find that the link graphs and UDGs are closely
related, in that pairs of graphs of one type sandwich graphs of the other type.

\begin{lemma}
For any $q\ge 1$ and any nearly-equilength linkset $L$, 
$U_{q-1}(L) \subseteq G_{q}(L)$ and $G_q(L) \subseteq U_{2(q+1)}(L)$.
\label{lem:graph-rels}
\end{lemma}

\begin{proof}
Recall that the links have lengths in the range $[d, 2d)$.
Let $\ell_v$ and $\ell_w$ be links that are neighbors in $U_{q-1}(L)$.
Then, $d(s_v, s_w) < (q-1) \cdot d$, by definition.
Thus, $d_{vw} \le d(s_v,s_w) + \ell_w < q \ell_w$,
and similarly $d_{wv} < q \ell_v$. 
Hence, $d_{vw} \cdot d_{wv} < q^2 \ell_v \ell_w$, so
$\ell_v$ and $\ell_w$ are neighbors in $G_{q}$. 

On the other hand, suppose we have neighbors $\ell_u$ and $\ell_w$ in $G_q$.
Notice that $d(s_u,s_w) \le d_{uw} + \ell_w < d_{uw} + 2d$,
and similarly $d(s_u, s_w) < d_{wu} + 2d$.
Then, 
\[ (d(s_u,s_w)-2d)^2 < d_{uw} \cdot d_{wu} < q^2 \ell_v \ell_w <
(2qd)^2\ . \]
Thus, $d(s_u,s_w) < 2(q+1)d$. Hence, $\ell_u$ and $\ell_w$ are
neighbors in $U_{2(q+1)}(L)$.
\end{proof}

Read differently, the above lemma implies that sender separation and
signal strength of a linkset go hand in hand. Namely, 
if $S$ is a $q$-independent set of nearly-equilength links, then the
senders in $S$ are of mutual distance at least $(q-1)d$, and thus
$U_{q-1}(S)$ is an empty graph (independent set).
Conversely, if $X$ is an independent set in unit-disc graph
$U_{2(q+1)}(L)$, where $L$ is nearly-equilength linkset,
then $X$ is $q$-independent.

We can now argue our claim that unit-disc graphs capture
well nearly-equilength links in fading metrics.
Define
   $z_2 = z_2(p) = p^{1/\alpha} - 1$.
\label{def:z2} 
We show that a linkset with minimum link length $d$ and pairwise sender
separation of at least $z_1(p)\cdot d$ will be a $p$-signal set, while any
$p$-signal set must obey a separation of at least $z_2(p) \cdot d$.

\begin{theorem}
For a set $L$ of nearly-equilength links,
any independent set in $U_{z_1}(L)$ is a $p$-signal linkset under
uniform power,
and any $p$-signal subset of $L$ is an independent set in $U_{z_2}(L)$.
\label{thm:nearequilengths-and-udg}
\end{theorem}

\begin{proof}
By Lemma \ref{lem:faraway-aff}, an independent set $X$ in
$U_{z_1}(L)$ is a $p$-signal set.
By Lemma \ref{lem:ind-separation}, a $p$-signal subset $S$ of $L$ is
$(p^{1/\alpha})$-independent (i.e., an independent set in $G_{p^{1/\alpha}}$).
By Lemma \ref{lem:graph-rels}, it is then an independent set in
$U_{p^{1/\alpha}-1}(L)$. 
\end{proof}

We note that unit-disc graphs in fading metrics satisfy a
\emph{bounded-independence} property as follows.
Recall that $\alpha(G)$ is the cardinality of a maximum independent
set in $G$.

\begin{observation}
Let $a$ and $b$ be given constants, $a \ge b$.
Let $U_a = U_{a}(L)$ and $U_b= U_{b}(L)$ be
unit-disc graphs on the same linkset $L$ but with different radii.
Let $\ell_v$ be a link in $L$, corresponding to a node $v$ in $U_a$ with closed neighborhood $N_1 = N_{U_a}[v]$.
Then, $\alpha(U_b[N_1]) \le C (1+2a/b)^A$.
\label{obs:clawfree-ness}
\end{observation}

\begin{proof}
The nodes in an independent set $I$ in $U_b$ form disjoint balls of
radius $bd/2$ centered at the senders of the links.  All senders of
links in $N_1$ are contained in the ball $B(s_v, ad)$, where $d$ is
the minimum link length in $L$.  Thus, all the balls corresponding to
$I$ are contained in the larger ball $B(s_v, (a+b/2)d)$.  The packing
constraint of the metric ensures that a limited number of the smaller
disjoint balls fit inside the large ball, implying that
$|I| \le \calP(B(s_v,(a+b/2)d), bd/2) \le C (1+2a/b)^A$.
\end{proof}

Our problems reduce then, within constant factors, to coloring and
(weighted) independent sets in UDGs. 
We say that an independent set in a weighted graph is \emph{greedy} if
it is obtained by the iterative process of selecting a vertex whose
weight is greater than each of its neighbors', deleting the neighbors,
and recursing on the remaining graph.  

The following result is immediate from Thm.~\ref{thm:nearequilengths-and-udg} and Obs.~\ref{obs:clawfree-ness}.

\begin{theorem}
Let $L$ be a nearly-equilength linkset.  Then, any maximal independent
set of $U_{z_1(\beta)}(L)$ is an $O(1)$-approximation of {\ssched} and
any greedy independent set of $U_{z_1(\beta)}(L)$ is an
$O(1)$-approximation of {\wcapacity}.
\label{thm:capacity-results}
\end{theorem}

We define a coloring of a graph $G$ to be \emph{minimal} if it uses at
most $D(G)+1$ colors, where $D(G)$ is the maximum degree of
a vertex in $G$. 

\begin{theorem}
Let $L$ be a nearly-equilength linkset.
Let $\calS$ be a minimal coloring of $U_{z_1}(L)$.
Then, using uniform power, 
$\calS$ induces a schedule that yields a $O(1)$-approximation to {\sched}.
\label{thm:udg-equilen}
\end{theorem}

\begin{proof}
The coloring $\calS$ forms an SINR-feasible schedule of $L$, by 
Thm.~\ref{thm:nearequilengths-and-udg}, and uses at most $D(U_{z_1}(L))+1$ colors, by
the minimality of the coloring.
Consider the closed neighborhood $N_1=N_{U_{z_1}}[v]$ of a maximum degree
node $v$ in $U_{z_1}$. By Obs.~\ref{obs:clawfree-ness},
at most $s = C (1+2z_1/z_2)^A$ nodes in $N_1$ can be in any feasible slot.
Hence, the optimal solution uses at least $|N_1|/s =
(D(U_{z_1}(L))+1)/s$ slots, for a performance ratio of $s$.
\end{proof}

The performance ratio of our algorithms is bounded by 
$C (1+\frac{2z_1}{z_2})^A$. 

Efficient distributed algorithms are known for coloring unit-disc
graphs in the plane \cite{Couture07} and more generally
bounded-independence graphs \cite{MoWa08}.
Thus, our characterization can be translated into distributed
constant-factor approximation algorithms of {\sched} and {\ssched} in
nearly-equilength linksets, when given the appropriate communication primitives.

\subsection{Scheduling Arbitrary Linksets}
\label{sec:arbitrary}

We can handle links of arbitrary lengths by partitioning them into
groups, where lengths of links in each group differ by a factor of at
most 2.  A simple approach is to schedule each group separately using
Thm.~\ref{thm:udg-equilen}, or to select
the largest of the approximately maximum (weighted) capacity subsets
from each of the groups.

Let $g(L) = |\{m: \exists \ell_v, \lceil \lg \ell_v \rceil = m \}|$
denote the \emph{length diversity} of the link set $L$, or the number
of length groups. Note that $g(L) \le \log \Delta$.

\begin{theorem}
The {\sched}, {\ssched}, and {\wcapacity} problems are $O(g(L))$-approximable,
using uniform power assignment.
\label{thm:gl}
\end{theorem}


Moscibroda and Wattenhofer \cite{MoWa06}
showed that uniform power scheduling can be highly suboptimal, and 
Moscibroda, Oswald and Wattenhofer \cite{MoscibrodaOW07}
showed that it can can be as much as a factor of $n$ or
$\Omega(\log \Delta)$ from optimal.  Specifically, they constructed a
set of links with the property that any there exists a power
assignment that makes the linkset feasible, while any use of uniform
power results in the trivial schedule of $n$ slots.
Hence, the ratio of $\theta(\log \Delta)$ is best possible for uniform power.
\smallskip

We can also claim easy \emph{online} algorithms.
The algorithm for {\sched} is in fact online.

\begin{corollary}
There is a deterministic online algorithm for {\sched} that is constant
competitive on nearly-equilength links and $O(\log
\Delta)$-competitive in general.
\label{cor:const-on-equilength}
\end{corollary}

A similar result can be attained by {\ssched} by a \emph{randomized}
online algorithm that randomly picks one of the length groups and then
picks greedily from that group.
If the value of $\Delta$ is not known, then an approach of Lipton and
Tomkins \cite{LiptonTomkins} can be used.

\begin{corollary}
There is a randomized $O(\log \Delta)$-competitive algorithm for
{\ssched}, when $\Delta$ is known in advance, and a
$O(\Delta^{1+\epsilon})$-competitive algorithm otherwise, for any $\epsilon > 0$.
\end{corollary}

\section{Approximations Using Mean Power}
\label{sec:oblivious}


We explore in this section the power of oblivious assignments.
The results of the preceding section apply to all oblivious power
functions, but are tight only for uniform and linear power assignments. 
We can greatly surpass these bounds by being selective about the
oblivious function used; in particular, we obtain these improvements
for the mean power assignment $\calM$.

We present in Sec.~\ref{sec:unidir} a scheduling algorithm using $\calM$
that achieves a ratio of $O(\log\log \Delta \cdot \log n)$. 
In the bidirectional setting, the algorithm obtains an
improved $O(\log n)$-ratio, as shown in Sec.~\ref{sec:bidi}.
The same results hold also for the (weighted) capacity problem.
We complement these results 
with a construction in Sec.~\ref{sec:construction} that
suggests a $\Omega(\log\log \Delta)$-separation between the lengths of
optimal schedules with or without oblivious power assignments.

We first introduce our approximation technique, which may be of
independent interest. This subsection may be skipped by the reader
that is not concerned with methods for weighted capacity or extensions
of the graph-theoretic notion of inductiveness.

\subsection{Approximation Via Inductiveness}

A common heuristic for subset problems is to find a ``good'' item, and
then recurse on the set of remaining items that are compatible with
the first one.  This yields good approximations if we can show that
only a small number of the items eliminated in each round can belong
to any optimal solution.  For instance, if the incompatibilities are
in the form of a graph and the set of nodes eliminated can be covered
by $k$ cliques, a $k$-approximation follows.  We generalize
this well-known concept to fit to our situation.

We have a set of links $L$, and a set property on $L$ in the form of
$p$-signal sets.
We shall partially capture this property with a
graph in the following sense. For a set of links to be feasible it
is a sufficient but not a necessary condition for it to form an
independent set in the graph. E.g., for set $L$ of nearly-equilength links, 
this property holds for the graph $U_{z_1}(L)$, by
Thm.~\ref{thm:nearequilengths-and-udg}. 
This motivates the extensions we put forth below.

A set property $\pi$ is said to be \emph{hereditary} if, whenever
$\pi(S)$ holds for a set $S$, it also holds for any $S' \subseteq S$.  In other words, $\pi$ represents a monotone Boolean function. A (sub)set
satisfying $\pi$ is said to be a $\pi$-(sub)set.  Let $\Pi(S)$ be the
maximum cardinality of a $\pi$-subset of $S$.  We say that a graph
$G=(V,E)$ is \emph{compatible} with a property $\pi$ on $V$ if any
independent set in $G$ satisfies $\pi$.

\begin{definition}
Let $V$ be a set of elements, $\pi$ be a hereditary set property defined on $V$, and $G=(V,E)$ be a graph on $V$ that is compatible with $\pi$. 
Then, $G$ is 
\emph{$k$-$\pi$-inductive} if there is an ordering $v_1, v_2, \ldots, v_n$ 
of the elements, such that for any $v_i$, $1 \le i \le n$, it
holds that $\Pi(N[v_i] \cap \{v_i,v_{i+1},\ldots, v_n\}) \le k$.
\label{def:inductive}
\end{definition}

To be useful in this context, property $\pi$ needs to be
polynomial-time checkable; this holds for the case of SINR feasibility
by solving a system of linear constraints.
Additionally, there needs to be an oracle to determine the
inductive ordering of the vertices.

This property generalizes the property of being \emph{sequentially
  $k$-independent}, where $\pi$ is the property of a vertex set being
independent in the graph. 
This latter property has been around for a while, but was first studied
explicitly in \cite{AkcogluADK2002}, followed by \cite{YeBorodin}.
Various optimization problems can be approximated on sequentially
$k$-independent graphs within a factor of $k$, including Weighted
Independent Set \cite{AkcogluADK2002} and Graph Coloring \cite{YeBorodin},
when an appropriate vertex ordering can be determined.

The \emph{(Weighted) Maximum $\pi$-subset} problem is defined as
follows for a given hereditary property $\pi$: Given a set $V$ of
items and a (vertex-weighted) graph $G=(V,E)$ compatible with $\pi$, find a maximum
(weight) subset $X \subset V$ that satisfies $\pi$. In the
\emph{Minimum Partition into $\pi$-Subsets} problem, we seek a
partition of $V$ into fewest number of $\pi$-subsets.
Note that optimal solutions to these problems do not depend on the graph $G$; 
rather, the structure of $G$ specifies the inductiveness characteristic.

\begin{proposition}
Let $\pi$ be a polynomially-time verifiable hereditary property with a
polynomial-time oracle to find $k$-$\pi$-inductive orderings.
Then, there are $k$-approximation algorithms for the \emph{Weighted
  Maximum $\pi$-Subset} and \emph{Minimum Partition into
  $\pi$-Subsets} problems on $k$-$\pi$-inductive instances.
\label{prop:approx-pi-probs}
\end{proposition}

\begin{proof}
Let $G=(V,E)$ denote an input instance, which by definition is
compatible with $\pi$.
To approximate the unweighted $\pi$-subset problem, we process the nodes 
in the $k$-$\pi$-inductive order.
For each vertex we encounter, we add it to our solution if it has no
neighbor among the previously added vertices.
This results in an independent set in $G$, which is a feasible
$\pi$-subset since $G$ is compatible with $\pi$.
For each node added to the solution, at most $k$ nodes 
from any feasible solution are eliminated from consideration, by
Def.~\ref{def:inductive}. 
Hence, our solution is within a factor of $k$ from optimal.

To approximate the partitioning problem, we process the nodes in the
reverse $k$-$\pi$-inductive order
and assign each node to the first class to which no neighbor in
$G$ has previously been assigned. Again, each set is independent and
thus we obtain a proper partition into $\pi$-sets. 
Let $v_i$ be a node assigned the
largest numbered class by this algorithm and let $N_i = \{v_j : (j > i) \wedge
(v_j, v_i) \in E(G)\}$ be the neighbors of $v_i$ that follow it in the
inductive order. Observe that the number of the class that $v_i$ is
assigned, and thus the total number of classes used by the algorithm,
is at most $|N_i|+1$.
On the other hand, by the definition of $k$-$\pi$-inductiveness, at most 
$k$ nodes in $N_i\cup\{v_i\}$ belong to any $\pi$-set, and thus the optimal
partition of $V$ uses at least $(|N_i|+1)/k$ classes. Hence, the algorithm
is $k$-approximate.

To approximate the weighted $\pi$-subset problem, we use the local
ratio algorithm of \cite{YeBorodin}. The algorithm and its proof are
given in the appendix for completeness.
\end{proof}

\subsection{Unidirectional Scheduling}
\label{sec:unidir}

In this subsection, which is the heart of the paper, we obtain
qualitatively improved link scheduling with oblivious power.

We shall utilize the \emph{mean} power assignment (or, square-root
assignment \cite{FKV09}) given by $\calM_v = \ell_v^{\alpha/2}$.
%
%
The affectance of link $\ell_w$ on link $\ell_v$ under $\calM$
is 
 \[ a_{w}(v) = \frac{\calM_{w}/d_{wv}^\alpha}{\calM_{v}/\ell_v^\alpha}
   = \left(\frac{\ell_w}{\ell_v}\right)^{\alpha/2} 
      \left(\frac{\ell_v}{d_{wv}}\right)^{\alpha}
   = \left(\frac{\sqrt{\ell_v \ell_w}}{d_{wv}}\right)^{\alpha} \ . \]
The following observation motivates the consideration of this power assignment.
\begin{observation}
Suppose $d_{wv} = d_{vw}$, for two links $\ell_v$, $\ell_w$.
Then, $a_{w}(v) = a_{v}(w)$ iff we use mean power assignment.
\end{observation}


Let $\tau = 2\beta n$ and $\Lambda = 2 \tau^{2/\alpha}$.
\label{def:tau} 
\label{def:lambda} 
We say that a link $\ell_v$ and $\ell_w$ are \emph{$t$-close} under mean
power assignment if, $\max(a_{v}(w),a_{w}(v)) \ge t$.
\label{def:tclose} 


The key observation that we make is that each link affects (or is
affected by) few links that are of widely different length.
We can then treat those affectance relationships in a graph-theoretic manner.
This central observation holds independent of metric.
Recall that we assumed that $\beta \ge 3^\alpha$, and thus it follows
from Lemma \ref{lem:ind-separation} that any slot in an optimal solution is a
3-independent linkset.

\begin{lemma}
Let $Q$ be a 3-independent set of links in an arbitrary
metric space, and let $\ell_v$ be a link that is shorter than the
links in $Q$ by a factor of at least $\Lambda$.
Suppose all the links in $Q$ are $\frac{1}{\tau}$-close to $\ell_v$
under mean power assignment.  Then, $|Q| = O(\log\log \Delta)$.
\label{lem:length-sep}
\end{lemma}

\begin{proof}
The set $Q$ consists of two types of links: those that affect $\ell_v$ by at
least $\frac{1}{\tau}$ under mean power, and those that are affected by $\ell_v$ by that amount.
We shall consider the former type; the argument is nearly identical
for the latter type, and will be omitted.

Consider a pair $\ell_w, \ell_{w'}$ in $Q$ that affect $\ell_v$ by at
least $1/\tau$, 
and suppose without loss of generality that $\ell_w \ge \ell_{w'}$.
The affectance of $\ell_w$ on $\ell_v$ implies that
$\sqrt{\ell_v \ell_w}^\alpha \ge d_{wv}^\alpha \cdot 1/\tau$, or 
\[ d_{wv} \le \sqrt{\ell_v \ell_w} \tau^{1/\alpha} 
         = \sqrt{\ell_v \ell_w \Lambda/2} \ . \]
Similarly, $d_{w'v} \le \sqrt{\ell_v \ell_{w'}\Lambda/2}$.
By the triangular inequality we have that 
\[ d_{w'w} \le d(s_{w'}, r_v) + d(r_v,s_w) + d(s_w, r_w) 
  \le \ell_{w} + \sqrt{2\Lambda \ell_v \ell_w} < 3\ell_w,  \]
using that $\sqrt{\ell_w} \ge \sqrt{\Lambda} \sqrt{\ell_v}$.
Similarly, 
\[ d_{ww'} \le d_{wv} + d_{w'v} + \ell_{w'} 
  \le \ell_{w'} + \sqrt{2 \Lambda \ell_v \ell_{w}} \ . \]
Multiplying together, we obtain that
\[ d_{w'w} \cdot d_{w w'} \le 3 \ell_{w'} \ell_w + 3 \sqrt{2\Lambda
     \ell_v \ell_{w}} \cdot \ell_w\ . \]
By 3-independence, 
$d_{w'w} \cdot d_{w w'} \ge 9 \ell_w \ell_{w'}$. 
By combining the last two inequalities and cancelling a $6 \ell_w$
factor, we have that 
$\ell_{w'} \le \sqrt{\Lambda \ell_v \ell_w/2}$, or
\begin{equation}
\ell_w \ge \frac{2 \ell_{w'}^2}{\Lambda \ell_v}\ .
\label{eq:two-length-rel}
\end{equation}

Label the links in $Q$ by $\ell_1, \ell_2, \ldots, \ell_t$ in 
increasing order of length.
Equation (\ref{eq:two-length-rel}) implies that 
\begin{equation}
\frac{\ell_{i+1}}{\ell_i} \ge \frac{2 \ell_i}{\ell_v \Lambda} \ge
2 \frac{\ell_i}{\ell_1},
\label{eq:lambda-ratios}
\end{equation}
for any $i=2, 3, \ldots, t$.
Thus, if we let $\lambda_i = \ell_i/\ell_1$, we get from 
(\ref{eq:lambda-ratios}) that $\lambda_{i+1} \ge 2 \lambda_{i}^2$, and by induction that 
$\lambda_{t} \ge 2^{2^{t-1}-1}$. Hence, $|Q| = t \le \lg\lg \lambda_t
+2 \le \lg\lg \Delta + 2$, and the lemma follows.
\end{proof}

We say that a set $S$ of links is \emph{well-separated} if any pair of
links differ in length by a factor that is either less than 2 or greater than
$\Lambda$, and that
a link $\ell_w$ is \emph{length-separated} from link $\ell_v$ if
$\ell_w > \Lambda \ell_v$.
\label{def:wellsep} 

We now proceed as follows.
We partition a given linkset $L$ into classes $L_1, L_2, \ldots,
L_{M}$, where $M = \lceil \lg 2 \Lambda \rceil$, such that 
$L_i = \{\ell_v : \exists k, \lceil \lg \ell_v \rceil = i + k M \}$. Namely,
each $L_i$ is a well-separated set. 
We shall solve the problems independently on the classes $L_i$ and
combine the subsolutions in the obvious way.

Let $S$ be a well-separated linkset and $d$ be the minimum link length
in $S$. Let $z = z_1(2^{1+\alpha/2}\beta)$.
Define the graph $H(S)$ on $S$ where two links $\ell_v$ and $\ell_w$
are adjacent if they are either:
a) nearly-equilength and the distance between their senders is at
most $z d$, or
b) length-separated and $1/\tau$-close.
We show the scheduling and capacity problems are captured well
as coloring and independent set problems on the graph $H$.

\begin{lemma}
Let $S$ be a well-separated linkset in a fading metric.
Then, any subset of $S$ that is 
independent in $H(S)$ is SINR-feasible using mean power.
\label{lem:is-is-feasible}
\end{lemma}

\begin{proof}
Let $X$ be a subset of $S$ that is independent in $H(S)$.
Consider a link $\ell_v$ in $X$. 
Let $S_v$ be the set of links in $S$ that are nearly-equilength to
$\ell_v$ (including $\ell_v$), $X_v = X \cap S_v$ and 
$\hat{X} = X \setminus X_v$. 
We bound the affectance on $\ell_v$ separately for $X_v$ and $\hat{X}$.
None of the links in $\hat{X}$ are $1/\tau$-close to $\ell_v$, so each
affects $\ell_v$ by at most $1/\tau$, for a total of
$a_{\hat{X}}(\ell_v) \le n\cdot 1/\tau = 1/(2\beta)$. 
By definition, $X_{v}$ is independent in $U_{z}(S_v)$, where $z = z_1(2^{1+\alpha/2}\beta)$,
and so by Thm~\ref{thm:nearequilengths-and-udg} it is a
$2^{1+\alpha/2}\beta$-signal set under uniform power.
Changing to mean power introduces a variance of at most $2^{\alpha/2}$
in the transmission powers, since the variance in length is at most 2.
Thus, $X_v$ is a $2\beta$-signal set under mean power.
Hence, under mean power, $a_{X_{v}}(\ell_v) \le 1/(2\beta)$ and
$a_{X}(\ell_v) = a_{X_v}(\ell_v) + a_{\hat{X}}(\ell_v) \le 1/\beta$.
\end{proof}

The graph $H(S)$ of a well-separated linkset $S$ has good
inductiveness properties.
Denote the case of $k$-$\pi$-inductiveness when $\pi$ refers to SINR
feasibility as \emph{$k$-SINR-inductive}.

\begin{lemma}
Let $S$ be a well-separated linkset in a fading metric.
Then, $H(S)$ is $O(\log\log \Delta)$-SINR-inductive.
The inductive ordering is that of non-decreasing link length.
\label{lem:clawlike-freeness}
\end{lemma}

\begin{proof}
Let $\ell_v$ be the shortest link in $S$. 
Let $X$ be an SINR-feasible subset of $N_H[\ell_v]$, the closed
neighborhood of $\ell_v$ in $H(S)$.
We shall show that $|X| = O(\log \log \Delta)$. We can then order
$\ell_v$ first and apply the claim inductively on $S
\setminus \{\ell_v\}$ to obtain the remainder of the inductive order,
yielding the lemma.  

Let $S_v$ be the subset of nearly-equilength links in $S$ of length at
most double that of $\ell_v$.
The nearly-equilength feasible ($\beta$-signal) linkset $X_v = X \cap S_v$ 
is an independent set in $U_{z_2(\beta)}(S_v)$, by
Thm.~\ref{thm:nearequilengths-and-udg}. 
Note that $X_v$ is contained in the closed-neighborhood of $\ell_v$ in $U_{z_1(2\beta)}(S_v)$.
Then, by Obs.~\ref{obs:clawfree-ness}, 
\[ |X_v| \le \alpha(U_{z_2(\beta)}[X]) \le C (1 + 2
z_1(2\beta)/z_2(\beta))^A = O(1)\ . \]

The other neighbors of $\ell_v$, those in $X \setminus X_v$, are
length-separated from $\ell_v$. By Lemma \ref{lem:length-sep},
$\ell_v$ has at most $O(\log\log \Delta)$ length-separated neighbors
in $X$.  Hence, $|X| = O(\log\log \Delta)+ O(1) = O(\log\log \Delta)$.
\end{proof}

We now apply Prop.~\ref{prop:approx-pi-probs} on each of the
$O(\log n)$ classes $L_i$ separately to obtain our main result.

\begin{theorem}
{\sched}, {\ssched}, and {\wcapacity} are 
$O(\log\log \Delta \cdot \log n)$-approximable in fading metrics. 
\end{theorem}

Finally, we obtain as corollary, a relationship between schedule length and the
chromatic number of a certain graph on the links.
Let $G'(L)$ be the graph on the linkset $L$ formed by the complete
union of the graphs $H(L_i)$, for $i=1, 2, \ldots$. Namely, links in
different length classes are adjacent in $G'$, while links in the same
length class $L_i$ induce the subgraph $H(L_i)$.

\begin{corollary}
There is an algorithm that outputs a feasible scheduling using
$O(\log\log \Delta \cdot \log n) \cdot \chi(G'(L))$ slots in fading metrics.
\end{corollary}

\subsection{Bidirectional Scheduling}
\label{sec:bidi}

In the bidirectional variant introduced by Fangh\"anel et al \cite{FKRV09},
a stronger separation criteria applies, since communication along each
link can occur in either direction. The asymmetry between sender and
receiver disappears and thus studying this model is useful in
order to explore the cost of assymetry.

The distance between two links is now the shortest distance between
any endpoints of the links. Thus, $d_{uv} = d_{vu} = \min(d(r_v,r_u),
d(r_v, s_u), d(s_v,s_u), d(s_v,r_u))$. Other definitions are unchanged.

We can obtain a better approximation ratio for this problem, with
essentially the same algorithm, via the following stronger version
of Lemma \ref{lem:length-sep}.

\begin{lemma}
Let $S$ be a set of $2$-independent links in a
bidirectional fading metric
and let $\ell_v$ be a link.
Then, there is at most one link $\ell_w$ in $S$ with $\ell_w >
\tau^{2/\alpha} \cdot \ell_v$ that is $1/\tau$-close under mean
power assignment.
\label{lem:single-long-aff}
\end{lemma}


\begin{proof}
Suppose the lemma is false and let $\ell_{w}, \ell_{w'}$ be two links in $S$
that are longer than $\tau^{2/\alpha}$ times $\ell_v$ and affect it by at least
$1/\tau$ each. Suppose without loss of generality that $\ell_w \ge \ell_{w'}$.
The assumption of affectance under mean power assignment implies that
\[ \left( \frac{\sqrt{\ell_v \ell_u}}{d_{vu}} \right)^\alpha \ge 1/\tau, \]
for $u \in \{w, w'\}$.
Thus, $d_{vu} \le \tau^{1/\alpha} \sqrt{\ell_v \ell_u}$.
In the bidirectional case, $d_{vu} = d_{uv}$.
Thus, by the triangular inequality, we have that
\[ d_{w' w} = d_{w w'} \le d_{wv} + d_{v w'} \le 2 \tau^{1/\alpha} \sqrt{\ell_v \ell_w} 
 < 2 \tau^{1/\alpha} \sqrt{(\ell_{w'}/\Lambda) \ell_w} 
 = \sqrt{2 \ell_{w'} \ell_w}\ . \]
Then, $\ell_w$ and $\ell_w'$ are not $2$-independent,
which contradicts our assumption.
\end{proof}

The rest of the argument is identical to the unidirectional case.
Lemma \ref{lem:is-is-feasible} still holds, while we get a stronger
version of Lemma \ref{lem:clawlike-freeness}.

\begin{lemma}
Let $S$ be a well-separated linkset in a bidirectional fading metric
and let $q$ be as in Lemma \ref{lem:is-is-feasible}.
Then, $H(S)$ is $O(1)$-SINR-inductive.
The inductive ordering is that of non-decreasing link length.
\label{lem:bidi-clawlike-freeness}
\end{lemma}

As before, we partition $L$ into $O(\log n)$ well-separated subsets
$L_1, L_2 \ldots$, using Prop.~\ref{prop:approx-pi-probs} on each of them.
This results in the following approximation results.

\begin{theorem}
There is an $O(\log n)$-approximation for the bidirectional
versions of {\sched}, {\ssched}, and {\wcapacity} in fading metrics.
\label{thm:bidi}
\end{theorem}

Finally, we get a tighter relationship with graphs.
Let $G'(L)$ be defined as in the previous subsection.

\begin{corollary}
There is an algorithm that outputs a feasible scheduling using
$O(\log n) \cdot \chi(G'(L))$ slots in fading metrics,
in the bidirectional setting.
\end{corollary}

\subsection{Construction}
\label{sec:construction}

We now give evidence that the upper bounds obtained are close to the best
possible for oblivious power functions.
A similar result follows also from the constructions in \cite{FKV09}
by analyzing the dependence on $\Delta$.

We say that a function $f$ is \emph{well-behaved} if there is an
$\epsilon > 0$, such that either a) for any $x > x'>0$, it holds that
$f(x) = O((x/x')^{\alpha-\epsilon})f(x')$ or
b) for any $x > x'>0$, it holds that
$f(x) = \Omega((x/x')^{\epsilon})f(x')$. 
Essentially, the definition stipulates that the function either grows
at a steady polynomial (possibly of very small degree) rate, or is
limited in its growth by a polynomial of degree strictly less than $\alpha$.
This means that the function can be jittery and locally unstable, but
on a large scale it can't be all over the place.
Intuitively, any reasonable power assignment function is well-behaved;
in particular, it holds for all functions considered in the literature,
which are polynomials.

\begin{theorem}
For any well-behaved power function $\phi$, 
there is a SINR-feasible instance for which any schedule under $\phi$
requires $\Omega(\log \log \Delta)$ slots.
\label{thm:construction}
\end{theorem}

\begin{proof}
Consider first the case when $\phi$ grows moderately slowly, 
i.e., there are fixed constants $\epsilon, c, c_0$ such that
for any $x, x'$ with $x > c_0 x'$,
$\phi(x) \le c \cdot x^{\alpha - \epsilon} \phi(x')$.
We assume for simplicity that $\beta = 1$.

Let $t= \max(\lceil (2\alpha+\lg c)/\epsilon\rceil, 4)$
and $c_1 = \lg\lg c_0$.
Consider the following set of links $L = \{\ell_1, \ell_2, \ldots, \ell_n \}$
located on the real line, where the length of link $\ell_i$ is $\ell_i = 2^{t^{i+c_1}}$.
Let $a_i = \sum_{j=0}^i \ell_j$, where $\ell_0$ denotes $2^{t^{c_1}}$.
Position the receiver $r_i$ of $\ell_i$ at location 
$+a_{i-1}$ and the sender $s_i$ at location $-(\ell_i - a_{i-1})$.
Observe that for any $i > j$, we have that
$\ell_i \ge c_0 \ell_j$, and thus
\begin{equation}
  \frac{\phi(\ell_i)}{\phi(\ell_j)} 
   \le c \left(\frac{\ell_i}{\ell_j}\right)^{\alpha-\epsilon}
   < c \ell_i^{\alpha-\epsilon} \ ,
\label{eqn:pow-ratio-bound}
\end{equation}
where the second inequality uses that $\ell_j > 1$.
Observe that for $i > j$, 
\[ d_{ji} = (\ell_{j} - a_{j-1}) + a_{i-1} 
  \le \ell_{i-1} - a_{i-2} + a_{i-1} 
   = 2 \ell_{i-1}, \]
and that
\[ \lg \frac{\ell_i^\epsilon}{\ell_{i-1}^\alpha} = t^{i+c_1} \epsilon - t^{i-1+c_1}\alpha
    \ge t\epsilon - \alpha \ge \lg c + \alpha, \]
which together imply that
\begin{equation}
  \ell_i^\epsilon \ge c 2^\alpha \ell_{i-1}^\alpha \ge c d_{ji}^\alpha \ .
\label{eqn:growth-bound}
\end{equation}
Thus, using Inequalities (\ref{eqn:pow-ratio-bound}) and
(\ref{eqn:growth-bound}), respectively,
we have that for $i > j$,
\[ a_j(i) = \frac{\phi_j}{\phi_i} \cdot \frac{\ell_i^\alpha}{d_{ji}^\alpha}
  > \frac{1}{c \ell_i^{\alpha-\epsilon}} \cdot \frac{\ell_i^\alpha}{d_{ji}^\alpha}
  = \frac{\ell_i^\epsilon}{c d_{ji}^\alpha} 
  \ge 1\ . \]
Hence, in any schedule based on the mean assignment, each of the $n$ links
must be assigned to distinct slots.

Consider instead the oblivious power assignment function
$\Psi(v) = \ell_v^\alpha/\log \ell_v$.
Note that for $i > j$ in the configuration above,
$d_{ji} = \ell_j - a_{j-1} + a_{i-1} > \ell_j$.
Then, under $\Psi$, we have that for $i > j$,
\[ a_{j}(i) = \frac{\Psi(\ell_j)/d_{ji}^\alpha}{\Psi(\ell_i)/\ell_i^\alpha} 
= \frac{\ell_j^{\alpha}}{d_{ji}^\alpha \log \ell_j}\cdot \log \ell_i
= \frac{\ell_j^{\alpha} t^{i-j} }{d_{ji}^\alpha}
\le t^{i-j}\ . \]
Note that for $k > i$, it holds that $d_{ki} = a_{i-1} + \ell_k - a_{k-1} \ge \ell_k/2$. 
Thus, for $k > i$,
\[ a_k(i) = \frac{\Psi(\ell_k)/d_{ki}^\alpha}{\Psi(\ell_i)/\ell_i^\alpha} 
 = \frac{\ell_k^{\alpha}}{d_{ki}^\alpha \log \ell_k}\cdot \log \ell_i
 = \frac{\ell_k^{\alpha} t^{i-k} }{d_{ki}^\alpha}
 \le 2 t^{i-k}\ . \]
It follows that under $\Psi$, for any link $\ell_i \in L$, it holds that
\[ a_{L}(i) \le \sum_{j < i} t^{j-i} + \sum_{k > i} 2t^{i-k}
     < 3 \sum_{k=1}^\infty t^{-k} 
      = \frac{3}{t-1} \le 1\ , \]
using that $t \ge 4$.
It follows that the linkset $L$ is SINR-feasible.
We thus obtain a lower bound on the performance ratio of any schedule
using $\phi$ of $n = \Omega(\log\log \Delta)$.

Consider now the complementary instance, where the direction
or the role of senders and receivers, has been reversed.
Then, 
nearly identical computation shows that any function that grows
no slower than $\Omega((x/x')^\epsilon)$ can also only schedule a single link in a
single slot. On the other hand, using power assignment $f(\ell_v) =
\lg \ell_v$, shows that the construction is SINR-feasible, 
giving the same $\Omega(\log\log \Delta)$ lower bound.

Finally, we can combine the two constructions into a single instance
that is hard to schedule for all well-behaved oblivious power functions, by taking
disjoint copies that are sufficiently separated in space.
\end{proof}

\section{Conclusions}

From a practical
perspective, it would be interesting if the logarithmic factor could
be removed, giving a $O(\log\log \Delta)$-approximation. Alternatively, 
non-oblivious power strategies that could be implemented in a distributed setting
would be highly desirable.

\subsection*{Acknowledgement}

I would like to extend sincere thanks to Tigran Tonoyan for kindly pointing
out an error in an earlier version and suggesting a correction.
I also thank Thomas Erlebach and Pradipta Mitra for helpful comments.
Finally, the existence of this paper owes much to Roger Wattenhofer for
introducing me to these fascinating questions.

\bibliographystyle{abbrv}
\bibliography{ref-jour}

\begin{thebibliography}{10}

\bibitem{AkcogluADK2002}
K.~Akcoglu, J.~Aspnes, B.~DasGupta, and M.-Y. Kao.
\newblock Opportunity-cost algorithms for combinatorial auctions.
\newblock In E.~J. Kontoghiorghes, B.~Rustem, and S.~Siokos, editors, {\em
  Applied Optimization 74: Computational Methods in Decision-Making, Economics
  and Finance}, pages 455--479. Kluwer Academic Publishers, 2002.

\bibitem{AD09}
M.~Andrews and M.~Dinitz.
\newblock Maximizing capacity in arbitrary wireless networks in the sinr model:
  Complexity and game theory.
\newblock In {\em Proceedings of the 28th Annual IEEE Conference on Computer
  Communications (INFOCOM)}, April 2009.

\bibitem{Assouad}
P.~Assouad.
\newblock Plongements lipschitziens dans $\mathbf{R}^n$.
\newblock {\em Soc. Math. France}, 111(4):429--448, 1983.

\bibitem{AulMosPenPer08}
V.~Auletta, L.~Moscardelli, P.~Penna, and G.~Persiano.
\newblock Interference games in wireless networks.
\newblock In {\em Proceedings of the 4th International Workshop On Internet And
  Network Economics (WINE)}, volume 5385 of {\em LNCS}, pages 278--285, 2008.

\bibitem{AEK08}
C.~Avin, Y.~Emek, E.~Kantor, Z.~Lotker, D.~Peleg, and L.~Roditty.
\newblock {SNR} diagrams: Towards algorithmically usable {SINR} models of
  wireless networks.
\newblock In {\em Proceedings of the Twenty-Eighth Annual ACM SIGACT-SIGOPS
  Symposium on Principles of Distributed Computing (PODC)}, 2008.

\bibitem{ALP09}
C.~Avin, Z.~Lotker, and Y.~A. Pignolet.
\newblock On the power of uniform power: Capacity of wireless networks with
  bounded resources.
\newblock In {\em Proceedings of the $\mathit{17}^{th}$ Annual European
  Symposium on Algorithms (ESA)}, 2009.

\bibitem{BrarBS08}
G.~S. Brar, D.~M. Blough, and P.~Santi.
\newblock The {SCREAM} approach for efficient distributed scheduling with
  physical interference in wireless mesh networks.
\newblock In {\em Proceedings of the 28th IEEE International Conference on
  Distributed Computing Systems (ICDCS 2008)}, pages 214--224, 2008.

\bibitem{chafekar07}
D.~Chafekar, V.~Kumar, M.~Marathe, S.~Parthasarathy, and A.~Srinivasan.
\newblock {Cross-layer Latency Minimization for Wireless Networks using SINR
  Constraints}.
\newblock In {\em Proceedings of the 8th ACM Int. Symposium on Mobile Ad Hoc
  Networking and Computing (MobiHoc)}, 2007.

\bibitem{Clarkson}
K.~Clarkson.
\newblock Nearest-neighbor searching and metric space dimensions.
\newblock In {\em Nearest-Neighbor Methods for Learning and Vision: Theory and
  Practice}. MIT Press, 2005.

\bibitem{Couture07}
M.~Couture, M.~Barbeau, P.~Bose, P.~Carmi, and E.~Kranakis.
\newblock Location oblivious distributed unit disk graph coloring.
\newblock In {\em Proceedings of the 14th International Conference on
  Structural Information and Communication Complexity (SIROCCO):}, pages
  222--233. Springer-Verlag, 2007.

\bibitem{ElBattE04journal}
T.~A. ElBatt and A.~Ephremides.
\newblock Joint scheduling and power control for wireless ad hoc networks.
\newblock {\em IEEE Transactions on Wireless Communications}, 3(1):74--85,
  2004.

\bibitem{ErlebachGrant}
T.~Erlebach and T.~Grant.
\newblock Scheduling multicast requests in the sinr model.
\newblock In {\em Proceedings of the Sixth International Workshop on
  Algorithmic Aspects of Wireless Sensor Networks (ALGOSENSORS)}, 2010.

\bibitem{FanghanelOnline2010}
A.~Fangh{\"a}nel, S.~Geulen, M.~Hoefer, and B.~V{\"o}cking.
\newblock Online capacity maximization in wireless networks.
\newblock In {\em Proceedings of the 22nd Annual ACM Symposium on Parallel
  Algorithms and Architectures (SPAA)}, pages 92--99, 2010.

\bibitem{FKRV09}
A.~Fangh\"anel, T.~Ke{\ss}elheim, H.~R\"acke, and B.~V\"ocking.
\newblock Oblivious interference scheduling.
\newblock In {\em Proceedings of the $\mathit{28}^{th}$ Annual ACM Symposium on
  Principles of Distributed Computing (PODC)}, August 2009.

\bibitem{FKV09}
A.~Fangh\"anel, T.~Ke{\ss}elheim, and B.~V\"ocking.
\newblock Improved algorithms for latency minimization in wireless networks.
\newblock In {\em Proceedings of the 37th International Conference on
  Algorithms, Languages and Programming (ICALP)}, July 2009.

\bibitem{Toth}
G.~Fejes~T{\'o}th.
\newblock {\em Lagerungen in der Ebene, auf der Kugel und in Raum}.
\newblock Springer-Verlag, Berlin, 2nd ed. edition, 1972.

\bibitem{gao08}
Y.~Gao, J.~C. Hou, and H.~Nguyen.
\newblock Topology control for maintaining network connectivity and maximizing
  network capacity under the physical model.
\newblock In {\em Proceedings of the 27th Annual IEEE Conference on Computer
  Communications (INFOCOM)}, 2008.

\bibitem{GHWW09}
O.~Goussevskaia, M.~M. Halld\'{o}rsson, R.~Wattenhofer, and E.~Welzl.
\newblock {Capacity of Arbitrary Wireless Networks}.
\newblock In {\em Proceedings of the 28th Annual IEEE Conference on Computer
  Communications (INFOCOM)}, April 2009.

\bibitem{Goussevskaia2008Local}
O.~Goussevskaia, T.~Moscibroda, and R.~Wattenhofer.
\newblock {Local Broadcasting in the Physical Interference Model}.
\newblock In {\em Proceedings of the Fifth ACM SIGACT-SIGOPS International
  Workshop on Foundations of Mobile Computing}, August 2008.

\bibitem{gouss2007}
O.~Goussevskaia, Y.~A. Oswald, and R.~Wattenhofer.
\newblock {Complexity in Geometric SINR}.
\newblock In {\em Proceedings of the Eigth ACM Int. Symposium on Mobile Ad Hoc
  Networking and Computing (MobiHoc)}, pages 100--109, 2007.

\bibitem{GronkMibiHoc01}
J.~Gr{\"o}nkvist and A.~Hansson.
\newblock {Comparison between graph-based and interference-based STDMA
  scheduling}.
\newblock In {\em Proceedings of the Second ACM Int. Symposium on Mobile Ad Hoc
  Networking and Computing (MobiHoc)}, pages 255--258, 2001.

\bibitem{kumar00}
P.~Gupta and P.~R. Kumar.
\newblock {The Capacity of Wireless Networks}.
\newblock {\em IEEE Trans. Information Theory}, 46(2):388--404, 2000.

\bibitem{us:esa}
M.~M. Halld\'{o}rsson.
\newblock Wireless scheduling with power control.
\newblock In {\em Proceedings of the $\mathit{17}^{th}$ European Symposium on
  Algorithms (ESA)}, September 2009.

\bibitem{HallMitra11}
M.~M. Halld\'{o}rsson and P.~Mitra.
\newblock {Wireless Capacity with Oblivious Power in General Metrics}.
\newblock In {\em Proceedings of the 22nd ACM-SIAM Symposium on Discrete
  Algorithms (SODA)}, January 2011.

\bibitem{HW09}
M.~M. Halld\'{o}rsson and R.~Wattenhofer.
\newblock {Wireless Communication is in APX}.
\newblock In {\em Proceedings of the 37th International Conference on
  Algorithms, Languages and Programming (ICALP)}, July 2009.

\bibitem{Heinonen}
J.~Heinonen.
\newblock {\em Lectures on Analysis in Metric Spaces}.
\newblock Springer, 1999.

\bibitem{KatzVW08}
B.~Katz, M.~V{\"o}lker, and D.~Wagner.
\newblock Link scheduling in local interference models.
\newblock In {\em Proceedings of the Fourth International Workshop on
  Algorithmic Aspects of Wireless Sensor Networks (ALGOSENSORS)}, pages 57--71,
  2008.

\bibitem{kesselheimSODA11}
T.~Kesselheim.
\newblock {A Constant-Factor Approximation for Wireless Capacity Maximization
  with Power Control in the SINR Model}.
\newblock In {\em Proceedings of the 22nd ACM-SIAM Symposium on Discrete
  Algorithms (SODA)}, January 2011.

\bibitem{LiptonTomkins}
R.~J. Lipton and A.~Tomkins.
\newblock Online interval scheduling.
\newblock In {\em Proceedings of the $\mathit{5}^{th}$ Annual ACM-SIAM
  Symposium on Discrete Algorithms (SODA)}, pages 302--311, 1994.

\bibitem{MaheshwariJD08}
R.~Maheshwari, S.~Jain, and S.~R. Das.
\newblock A measurement study of interference modeling and scheduling in
  low-power wireless networks.
\newblock In {\em Proceedings of the 6th International Conference on Embedded
  Networked Sensor Systems (SenSys)}, pages 141--154, 2008.

\bibitem{Moscibroda07}
T.~Moscibroda.
\newblock The worst-case capacity of wireless sensor networks.
\newblock In {\em Proceedings of the 6th International Conference on
  Information Processing in Sensor Networks (IPSN)}, pages 1--10, 2007.

\bibitem{MoscibrodaOW07}
T.~Moscibroda, Y.~A. Oswald, and R.~Wattenhofer.
\newblock How optimal are wireless scheduling protocols?
\newblock In {\em Proceedings of the 26th Annual IEEE Conference on Computer
  Communications (INFOCOM)}, pages 1433--1441, 2007.

\bibitem{MoWa06}
T.~Moscibroda and R.~Wattenhofer.
\newblock {The Complexity of Connectivity in Wireless Networks}.
\newblock In {\em Proceedings of the 25th Annual IEEE Conference on Computer
  Communications (INFOCOM)}, 2006.

\bibitem{MoWa08}
T.~Moscibroda and R.~Wattenhofer.
\newblock {Coloring Unstructured Radio Networks}.
\newblock {\em Distributed Computing}, 21:271--284, 2008.

\bibitem{Moscibroda2006Protocol}
T.~Moscibroda, R.~Wattenhofer, and Y.~Weber.
\newblock {Protocol Design Beyond Graph-Based Models}.
\newblock In {\em Proceedings of the Fifth Workshop on Hot Topics in Networks
  (HotNets)}, November 2006.

\bibitem{moscibroda06b}
T.~Moscibroda, R.~Wattenhofer, and A.~Zollinger.
\newblock {Topology Control meets SINR: The Scheduling Complexity of Arbitrary
  Topologies}.
\newblock In {\em Proceedings of the $\mathit{6}^{th}$ ACM Int. Symposium on
  Mobile Ad Hoc Networking and Computing (MobiHoc)}, 2006.

\bibitem{ScheidelerRS08}
C.~Scheideler, A.~W. Richa, and P.~Santi.
\newblock An {$O(\log n)$} dominating set protocol for wireless ad-hoc networks
  under the physical interference model.
\newblock In {\em Proceedings of the $\mathit{9}^{th}$ ACM Int. Symposium on
  Mobile Ad Hoc Networking and Computing (MobiHoc)}, pages 91--100, 2008.

\bibitem{Weisstein}
E.~W. Weisstein.
\newblock Circle packing.
\newblock From MathWorld--A Wolfram Web Resource.
  \url{http://mathworld.wolfram.com/CirclePacking.html}.

\bibitem{YeBorodin}
Y.~Ye and A.~Borodin.
\newblock Elimination graphs.
\newblock In {\em Proceedings of the 37th International Conference on
  Algorithms, Languages and Programming (ICALP)}, pages 774--785, 2009.

\end{thebibliography}

\appendix
\section{Approximating Weighted Capacity on $k$-$\pi$-Inductive Graphs}

We apply the algorithm of Ye and Borodin \cite{YeBorodin} for
Weighted Independent Set in sequentially $k$-independent graphs to
the Weighted Maximum $\Pi$-subgraph problem
in $k$-$\pi$-inductive graphs.

\begin{theorem}{\cite{YeBorodin}}
Let $\pi$ be a polynomially-time verifiable property such that
any independent set satisfies $\pi$.
Then, the Weighted Maximum $\pi$-Subgraph problem is $k$-approximable
on graphs that are $k$-$\pi$-inductive.
\end{theorem}

\begin{proof}
Let $G=(V,E)$ be a $k$-$\pi$-inductive graph with weight function
$w:V\rightarrow \mathbf{R}$.
The algorithm maintains a stack $S$ of nodes, first pushing the nodes
onto the stack, and then popping them off.
\begin{pseudocode}
Let $v_1,v_2,\ldots,v_n$ be the nodes in the $k$-$\pi$-inductive order.
Initialize $\hat{w}(v_i)\leftarrow{}w(v_i)$
for $i\leftarrow{}1$ to $n$ do       // \emph{Push phase}
  if ($\hat{w}(v_i)>0$)
     push $v_i$ on $S$
     for each neighbor $v_j\in{}N(v_i)\cap{}\{v_{i+1},\ldots,v_n\}$ do
        Subtract $\hat{w}(v_i)$ from $\hat{w}(v_j)$
$A\leftarrow\emptyset$
while ($S$ is not empty) do   // \emph{Pop phase}
  $u\leftarrow$pop($S$)
  if ($u\cup{}A$ is a $\pi$-set)
    add $u$ to $A$
output $A$
\end{pseudocode}
Let $A$ be the output of the algorithm and $O$ be an optimal
solution. Let $S$ be the set of vertices in the stack at the end of the
push phase and $S_i$ be the contents of the stack when $v_i$ is being
considered in the push phase. 

We first prove that 
the stack algorithm achieves at least the total weight of the stack.
For a node $v_i$, let $\bar{w}(v_i)$ 
denote the final value of $\hat{w}(v_i)$, which it attains before
iteration $i$.
Then, it holds for each node $v_i$ that
\begin{equation}
 w(v_i) = \bar{w}(v_i) + \sum_{v_j \in S_i \cap N(v_i)} \bar{w}(v_j)\ .
\label{eqn:barw}
\end{equation}
If we sum up for all $v_i \in A$, we have
\begin{equation}
\sum_{v_i \in A} w(v_i) = \sum_{v_i \in A} \bar{w}(v_i) + \sum_{v_i
    \in A} \sum_{v_j \in S_i \cap N(v_i)} \bar{w}(v_j) 
  \ge \sum_{v_t \in S} \bar{w}(v_t),
\label{eqn:a-sbar}
\end{equation}
where the second equality holds because for any $v_t \in S$, we either have
$v_t \in S_i \cap N(v_i)$ for some $v_i \in A$, or we have $v_t \in A$.

Now we prove that the optimal solution achieves at most $k$ times the
weight of the stack. 
If we sum Equation (\ref{eqn:barw}) up for all $v_i \in O$, we have
\begin{equation}
\sum_{v_i \in O} w(v_i) \le \sum_{v_i \in O} \bar{w}(v_i) + \sum_{v_i
    \in O} \sum_{v_j \in S_i \cap N(v_i)} \bar{w}(v_j) 
  \le k \sum_{v_t \in S} \bar{w}(v_t),
\label{eqn:o-sbar}
\end{equation}
where the second inequality holds because when we sum up for all $v_i \in
O$, each of the terms $\bar{w}(v_t)$ for any vertex $v_t \in S$ can appear 
at most $k$ times, since the ordering $v_1, v_2, \ldots, v_n$ is a
$k$-$\pi$-inductive ordering. 
Combining (\ref{eqn:a-sbar}) and (\ref{eqn:o-sbar}), we have
\[ \sum_{v_i \in O} w(v_i) \le k \sum_{v_i \in A} w(v_i) \ . \]
\end{proof}

\end{document}